\documentclass[11pt]{article}

\usepackage{sn-preamble} 
\usepackage{tikz}
\usepackage{verbatim}
\usepackage{cancel}
\usepackage{enumitem}
\usepackage{bm}
\usepackage{caption}
\usepackage{changepage}   

\newcommand{\ignore}[1]{}

\newcommand{\ALGWARMUP}{\textsf{ALG-WARMUP}\ }

\def\bz{{\boldsymbol{z}}}
\def\bL{{\boldsymbol{L}}}

\DeclareMathOperator{\arccosh}{arccosh}

\title{Learning Depth-$3$ Circuits with Polynomial Savings}
\author{Xi Chen\thanks{Email: \texttt{xichen@cs.columbia.edu}}
\\ \textsl{Columbia University} \and 
Animesh Fatehpuria \thanks{Email: \texttt{animesh@anthropic.com}} \\ \textsl{Anthropic} \and 
Shyamal Patel\thanks{Email: \texttt{shyamalpatelb@gmail.com}} \\ \textsl{UT Austin} \and Rocco A. Servedio\thanks{Email: \texttt{ras2105@columbia.edu}}\\ \textsl{Columbia University}}
\date{}

\begin{document}

\pagenumbering{gobble}
\maketitle
\begin{abstract}
We study the challenging problem of learning depth-three circuits in the mistake-bound model of (realizable) online learning,  which is a more difficult model than distribution-free PAC learning.
Prior algorithms for this problem, due to Servedio and Tan \cite{ST17itcs}, could only learn polynomial-size depth-three circuits of $\poly(n)$ size over $\{0,1\}^n$ with a running time of $2^{n - \Omega(n/\log n)}$, and hence they ran in time $N^{1-o(1)}$ where $N=2^n$ is the running time of a naive memorization-based approach.

In this work we substantially improve on the \cite{ST17itcs} result:  for any constant $\gamma\geq1$, we give an algorithm that learns depth-three circuits of size $n^\gamma$ with running time 
\[
2^{n-c_\gamma n},
\] 
where $c_\gamma>0$ depends only on $\gamma$ and not on $n$. Hence we achieve a polynomial savings over the naive approach for learning any polynomial-size depth-three circuit.

The main driving force behind our improvement is an improved bound on the approximate degree of width-$k$ CNFs. Inspired by Szegedy \cite{Szegedy04} and Magniez et al.~\cite{MNRS11}, the rough idea of our construction is to use a Chebyshev polynomial to efficiently amplify the spectral gap of a carefully designed random walk.
This is combined with a random-restriction-like approach to separately learn different subfunctions corresponding to different assignments to a randomly chosen set of variables, using the Perceptron algorithm over a specially designed feature space.
 A simplified warmup instantiation of our approach achieves $c_\gamma = \exp(-O(\gamma))$; by augmenting this warmup with further ingredients we obtain the sharp form of our result, which achieves $c_\gamma=\Omega(1)/\gamma$.
\end{abstract}

\newpage


\newpage

\pagenumbering{arabic}

\section{Introduction} \label{sec:intro}

For more than forty years, dating back to the first days of Valiant's PAC learning model \cite{Valiant:84}, a major thrust of research at the intersection of computational learning theory and computational complexity has been the quest to learn rich and  expressive classes of Boolean functions.
This effort has proceeded across many different learning models which incorporate many different assumptions about the data distribution and the allowed forms of access to the unknown target function. 
It has contributed to, and been enriched by, the study of a wide range of structural properties of different kinds of Boolean functions (see e.g.~\cite{EhrenfeuchtHaussler:89,blu92,LMN:93,Mansour:95,BshoutyTamon:96,BBB+:00,KS04,KOS:04,MOS04,Lee09formulas,ACRSZ10,Kane:Gotsman11,Kane:sens14,Sherstov13sicomp,Sherstov13combinatorica,Peres21NoiseStability} and many other works).  
It continues to be an active research topic to the present day; see \cite{Servedio25PAC} for a fairly comprehensive survey, as well as \cite{DIKZ25,ANPS25,APS26,CKSV26} for some representative papers of this sort that have appeared since the recent (at the time of this writing) appearance of that survey.

But alas, it is very difficult to learn rich and expressive Boolean functions.  This is especially true in the original (distribution-free) PAC learning model introduced by Valiant in \cite{Valiant:84} and in related models which are known to be at least as difficult as the distribution-free PAC model, such as the \emph{online mistake-bound} learning model \cite{lit88} which is the subject of this paper (we formally define this model in \Cref{sec:model}).  This is particularly unfortunate because the online mistake-bound model is an attractive target of study: it is simple, elegant, and avoids the strong and potentially unrealistic distributional assumptions (uniform distribution over $\zo^n$, Gaussian distribution over $\R^n$, etc.)
that characterize a large body of work giving efficient, or at least sub-exponential, learning algorithms for many types of Boolean functions.\footnote{See e.g.~\cite{Verbeurgt:90,GKS93,Schapire94,HancockMansour:91,Hancock:93,kusman93,LMN:93,GHM:96,BlumKannan:97,Vempala:10,KOS:04,MOS04,Val15,BshoutyTamon:96,OW13,Servedio:04iandc,GKK:08,KK09,Feldman10,Feldman12b,OdonnellServedio:07,KKMS:08,KOS:08,GKM12colt,CIKK16,DKKTZ23}, among many others, for papers that learn different types of Boolean functions under such distributional assumptions.} 

The difficulty of distribution-free PAC learning or online mistake-bound learning of complex functions is perhaps most acutely evident when we measure complexity through the lens of \emph{circuit depth}.  A number of sub-exponential time algorithms have been developed for learning different kinds of depth-two circuits.  For example, in the distribution-free PAC learning model, $2^{\tilde{O}(n^{1/3})}$-time algorithms have been known for polynomial-size DNF or CNF formulas over $\zo^n$ for more than two decades \cite{KS04}; $2^{\tilde{O}(\sqrt{n})}$-time algorithms have been known for nearly as long for polynomial-size $\F_2$-polynomials (equivalently, $\mathsf{Parity} \circ \mathsf{And}$ circuits) \cite{HellersteinServedio:07}; and recent work has given $2^{\tilde{O}(\sqrt{n})}$-time algorithms for depth-two circuit classes such as $\mathsf{Any}_{\ell(n)}  \circ \mathsf{LTF}$ for $\ell(n) = \tilde{O}(\log n)$ \cite{APS26} and intersections of polynomially many polynomial-weight halfspaces over $\zo^n$ \cite{PV26}.
But despite intensive research efforts, and the existence of quasipolynomial-time algorithms in the uniform-distribution PAC model, no $2^{o(n)}$-time distribution-free PAC or online mistake-bound learning algorithms are known for even the simplest depth-3 circuit class of polynomial-size de Morgan (AND/OR) circuits. 

Motivated by the difficulty of obtaining subexponential-time learning for rich circuit classes, \cite{ST17itcs} suggested a change of perspective.  That work introduced the research goal of \emph{learning with non-trivial savings}, i.e.~obtaining learning algorithms for functions over $\zo^n$ that run in time $2^{n-s(n)}$ for some \emph{savings function} $s(n)$ which is as large as possible.
\cite{ST17itcs} highlighted connections with other circuit analysis problems such as compression \cite{CKKSZ} and gave a number of preliminary results on learning different circuit classes with non-trivial savings. In particular, for polynomial-size depth-three circuits, \cite{ST17itcs} gave a (randomized) learning algorithm that learns in time $2^{n - n/\Omega(\log n)}$ (see their Theorem~3).  
A number of other works have explored different aspects of the learning-with-nontrivial-savings framework \cite{OliveiraSanthanam17,Karchmer24}, but none improved on the above running time for depth-three circuits.

\medskip
\noindent {\bf Prior work: The approach of \cite{ST17itcs}.}
Since our algorithmic approach to learning depth-three circuits  builds on that of \cite{ST17itcs}, let us give a high-level description of their algorithm and its analysis, slightly simplifying certain aspects to ease the presentation.  At a high level, the \cite{ST17itcs} algorithm works as follows (we will refine this description a bit in the subsequent discussion):

\begin{enumerate}

\item Draw a random subset $\bL \subseteq [n]$ of ``live'' variables from a suitable distribution.  Call an assignment $z \in \zo^{[n] \setminus \bL}$ fixing all other variables an \emph{address}. Let $z \circ \zo^{\bL}$ denote the subcube obtained by fixing the variables outside $\bL$ to the address $z$.

\item For each address $z \in \zo^{[n] \setminus \bL}$, run a separate copy of an algorithm $A$ that  learns polynomial threshold functions of degree $d$ in the online mistake-bound model\footnote{It is well known that using linear threshold function learning approaches, such algorithms are known that have running time $\poly\pbra{{m \choose \leq d}}$ when run over $\zo^m$, see e.g. \cite{MT:94}.}  over the examples belonging to the subcube $z \circ \zo^{\bL}$. (This is done by routing each counterexample $x \in \zo^n$ that is received by the online learner to the copy of $A$ that is running over the subcube with the same address as $x$.)

\end{enumerate}

\noindent
The idea of the analysis is simple and is based on switching lemmas from low-level circuit complexity theory.
Let $F$ be the depth-three target circuit, and suppose without loss of generality that $F = g_1 \vee \cdots \vee g_s$ where each $g_i$ is a $\poly(n)$-size CNF formula. (The case in which the top gate is an AND function is handled by essentially the same arguments, switching the roles of AND and OR throughout. For the rest of the paper we will suppose that the depth-three target circuit is an $\mathsf{Or}\circ\mathsf{And}\circ\mathsf{Or}$ circuit.)  
For a random choice of $\bL$ and a typical (i.e.~random) choice of the address $\bz \in \zo^{[n]\setminus \bL}$, considering $F$ on the subcube $\bz \circ \zo^{\bL}$ is precisely the same as applying a \emph{random restriction} to $F$; we denote this restricted function over $\bz \circ \zo^{\bL}$ by $F|_{\bL,\bz}$.  
The well-developed theory of switching lemmas tells us that ``random restrictions simplify small-depth circuits''; in more detail, with high probability each restricted CNF $g_1|_{\bL,\bz},\dots,g_s|_{\bL,\bz}$ ``collapses'' to a depth-$d$ decision tree, which in particular is a degree-$d$ polynomial. 
It follows that under such an $(\bL,\bz)$ pair, the restriction $F|_{\bL,\bz}$ is a polynomial threshold function of degree at most $d$.  
Letting $M$ denote the mistake bound of algorithm $A$ when it is run on a degree-$d$ polynomial threshold function over $|\bL|$ variables, we get that for each subcube $\bz \circ \zo^{\bL}$ such that the high-probability simplification event indeed takes place, the corresponding copy of $A$ will make at most $M$ mistakes.  
(There will also be subcubes where the low-probability failure-to-simplify event takes place; the algorithm handles these subcubes simply by aborting the execution of $A$ within each subcube once that execution has made more than $M$ mistakes, and switching to a memorization-based approach that runs in time $2^{|\bL|}$ over that subcube.\footnote{This is the refinement that was alluded to above.})

The detailed analysis of \cite{ST17itcs} involves choosing $\bL$ from the right distribution as well as trading off the failure probability of simplification under random restrictions, using recent sophisticated ``multi-switching lemmas'' \cite{Hastad2014CorrelationParity}, against the running time of the polynomial threshold function learning routine. Intuitively, choosing a larger value of $d$ decreases the probability of the switching lemma failure event (failure on a subcube is expensive since it involves paying ``full fare'' with a $2^{|\bL|}$ running time for the memorization-based approach on that subcube) but increases the $\poly\pbra{{|\bL| \choose \leq d}}$ running time for typical subcubes where the switching lemma failure event does not take place. Optimizing these tradeoffs leads to choosing a random subset $\bL \subseteq [n]$ by including each variable in $\bL$ with probability $p=1/O(\log n)$, and leads to the overall $2^{n - n/O(\log n)}$ running time that was stated earlier.

\subsection{Our main result: Improved running time for learning depth-three circuits.}
The main result of this paper is a faster algorithm than that of \cite{ST17itcs} for learning depth-three circuits in the online mistake-bound model (the same model considered by \cite{ST17itcs}).   Informally, for any constant $\gamma \geq 1$, we give an algorithm that learns depth-three circuits of size $n^\gamma$ with running time $2^{n-c_\gamma n}$. Hence, as stated in the abstract, 
we achieve a \emph{polynomial} savings over the naive approach for any polynomial-size target depth-three circuit, though the polynomial savings that is achieved grows small as the constant exponent of the polynomial-size depth-three circuit grows large.

Let us state our main result in more detail:

\begin{theorem}
\label{thm:main}
For every constant $\gamma\ge1$ there is a constant $c_\gamma > 0$ and a randomized online mistake-bound learning algorithm such that for any target function $F: \zo^n \to \zo$ that is computed by a depth-three circuit of size at most $n^\gamma$, with probability at least $1-O(1/n)$ the algorithm learns $F$ in time at most 
\[
2^{n - c_\gamma n}.
\]
\end{theorem}
\begin{remark}
We prove two versions of \Cref{thm:main}.  First, as an entirely self-contained warmup, we prove the theorem with $c_\gamma = 2^{-O(\gamma)}$; this is done in \Cref{sec:warmup}.  Then, by augmenting the ideas and algorithmic approaches of the warmup with several new technical ingredients, in \Cref{sec:real} we prove the theorem with $c_\gamma = \Omega(1/\gamma)$.
\end{remark}

At a high level our algorithm is structured quite similarly to the \cite{ST17itcs} algorithm: like \cite{ST17itcs} we first draw a random subset $\bL$ of variables and then run a separate copy of a learning algorithm for each subcube $z \circ \zo^{\bL}$.  However, the \emph{analysis} of our algorithm is quite different from that of \cite{ST17itcs}: we do not use any switching lemma results from the literature (that role is instead played by a one-page self-contained calculation), and we crucially rely on \emph{approximating polynomials} and (for our main result) a novel \emph{feature representation}, two ingredients which are not present in any form in the \cite{ST17itcs} analysis.

In the rest of this introduction we give an initial overview both of our warmup $2^{n - 2^{-O(\gamma)}n}$-time algorithm and of our main $2^{n - \Omega(1/\gamma)n}$-time algorithm, and highlight some of the key tools underlying the new results and how the new result differs from \cite{ST17itcs}.
The full technical overview (see \Cref{sec:tech-overview}) will give additional details and intuition, and then the complete proofs will be given in \Cref{sec:prelims} and beyond.

\medskip
\noindent {\bf The central tool:  Improved approximating polynomials for width-$k$ CNF formulas.}
The technical heart of our improvement over \cite{ST17itcs} comes from a new construction of \emph{pointwise approximating polynomials} for width-$k$ CNF formulas.  Let us recall the previous state of the art here, which is due to Sherstov \cite{Sherstov20}:  Theorem~5.1 of  \cite{Sherstov20} states that if $f: \zo^\ell \to \zo$ is computed by a width-$k$ DNF or CNF formula, then for $\eps \in (0,1/2)$ there is a real polynomial $p: \zo^\ell \to \R$ that pointwise approximates $f$ to an additive $\pm \eps$ error on every $\zo^\ell$ input and has degree
\begin{equation}
\deg(p) \leq O\pbra{2^{k/2} \ell^{k/(k+1)} \log^{1/(k+1)} (1/\eps)}.
\label{eq:sherstov}
\end{equation}
Inspecting the degree bound given by \Cref{eq:sherstov}, we see that because of the $2^{k/2}$ factor the  bound is trivial (larger than $\ell$) for any value $k = \Theta(\log \ell)$. Unfortunately, $k=\Theta(\log \ell)$ is precisely what can be achieved by applying known switching lemmas to unrestricted polynomial-size CNF formulas; and as sketched above, such switching lemmas are central to the proof approach of \cite{ST17itcs}.  This is why the analysis of \cite{ST17itcs} cannot make use of the \cite{Sherstov20} results.

One of our main contributions is a new approximate degree bound for width-$k$ CNFs (or DNFs). We give intuition for this new bound in \Cref{sec:tech-overview-apx-degree} and the new result is stated in detail and proved in \Cref{thm:apx-degree}. Its statement is essentially the same as the result of \cite{Sherstov20} stated above, but now with degree bound
\begin{equation}
\deg(p) \leq O\pbra{\ell^{k/k+1} \log^{2/k+1}(1/\eps)}.
\label{eq:us}
\end{equation}
While this bound is slightly worse as a function of $\eps$, the difference is small (only a constant factor) in the $k=\Theta(\log \ell), \eps = 1/\poly(\ell)$ regime; moreover, crucially the new bound no longer has the $2^{k/2}$ factor, which is a considerable one in this regime. 
In particular, inspecting \Cref{eq:us}, we see that with $\eps = \poly(1/\ell)$ it gives a nontrivial bound even for a suitable choice of $k=\Theta(\log \ell)$.  
 
Briefly, the new construction of approximating polynomials treats evaluating a width-$k$ CNF on a given assignment $x$ as searching for a small subset of the input variables that contains a clause which is falsified by $x$.  This search is carried out via a random walk over such subsets whose progress is amplified by a Chebyshev polynomial, in the spirit of the quantum walk search algorithms of Szegedy and Magniez et al.~\cite{Szegedy04,MNRS11}. Balancing the size of the subsets against the length of the walk gives the $\ell^{k/(k+1)}$ exponent; a fuller explanation is given in \Cref{sec:tech-overview-apx-degree}.

\medskip

\noindent {\bf The warmup.}  Having a nontrivial bound in \Cref{eq:us} even when $k=\Theta(\log \ell)$ is  helpful for us for the following reason.  Under a typical outcome of the live variable set $\bL$, we can ensure that for almost all addresses $z \in \zo^{[n]\setminus \bL}$  as in the \cite{ST17itcs} approach, each depth-two CNF $g_i$ simplifies under the restriction $(\bL,z)$ in a useful way.  A detailed description of the ``useful way'' in which the CNFs simplify is somewhat involved and is deferred to \Cref{sec:tech-overview}, but roughly speaking, the idea is that given any assignment to a typically-not-too-large ``width-reducing set'' $Y(\bL,z)$ of variables (which is a subset of $\bL$), the restricted CNFs $g_i|_{\bL,z}$ have width at most $k$ where $k=\Theta(\log |\bL|)$.  Since our new polynomial approximation theorem, \Cref{thm:apx-degree}, gives us that each such width-$k$ CNF has a high-accuracy approximating polynomial of non-trivially small degree, it follows that the restricted target function $F|_{\bL,z}$ (which is simply an OR over the width-$k$ restricted CNFs)  has a low-degree polynomial threshold function representation for a typical assignment $z$.  For each such assignment we use the Perceptron algorithm  over an expanded feature space of all low-degree monomials to learn the target function under the restriction corresponding to that assignment (and as in \cite{ST17itcs} we pay full fare with a memorization-based approach for the few exceptional assignments).
This concludes our initial high-level description of our warmup result with its $2^{n - 2^{-O(\gamma)}n}$ running time.

\medskip
\noindent {\bf The full result.}
Our full result, with its $2^{n - \Omega(1/\gamma)n}$ running time, is achieved by  augmenting the above approach with several additional ingredients.  One new ingredient is based on a scheme of grouping the variables in $\bL$ into ``blocks.'' This lets us go beyond approximation by low-degree polynomials and instead construct a low-weight linear form, over a carefully structured set of ``block indicator'' features, which is a pointwise high-accuracy approximator for any width-$k$ CNF.  As in the warmup, given any assignment to the ``width-reducing set'' $Y(\bL,z) \subseteq \bL$, the overall function under restriction $(\bL,z)$ becomes an OR of width-$k$ CNFs, and hence is expressible as a low-weight linear threshold function over the block indicator features.  Another new ingredient that plays a crucial role is an exponential moment bound on the size of the ``width-reducing set'' $Y(\bL,\bz)$ mentioned above. 
With these modifications to the analysis, the algorithm again runs many copies of the Perceptron algorithm --- one for each assignment to the non-live variables --- but now over the new feature space of block indicators rather than over the feature space of low-degree monomials.  A careful analysis of the aggregate performance of these $2^{n-|\bL|}$ copies of Perceptron, using the exponential moment bound and the existence of the low-weight linear threshold functions (over block indicator features) for most restrictions $F|_{\bL,z}$ of $F$, gives us the full $2^{n - \Omega(1/\gamma)n}$ result.

\medskip
\textbf{Statement on AI Use.} The main ideas in this paper were proposed by Anthropic's Fable 5 model. The authors additionally used Sol and Fable to aid in determining how to present the material and check for typos and errors in the manuscript. The human authors studied, verified and streamlined these ideas and wrote the paper adding motivation and exposition and substantially reorganizing the technical material. The final paper reflects the understanding of the human authors, who take full responsibility for all of the technical content, expository content, and references in the paper.


\section{Technical Overview} \label{sec:tech-overview}

\subsection{The new approximate degree bound for width-$k$ CNF formulas (\Cref{thm:apx-degree})}
\label{sec:tech-overview-apx-degree}

{\bf Background and contrast with prior work:}  As mentioned earlier, prior to our work Sherstov \cite{Sherstov20} had shown that every width-$k$ CNF formula over $\zo^\ell$ has a pointwise $\eps$-approximating polynomial of degree $O(2^{k/2} \ell^{k/(k+1)} \log^{1/(k+1)} (1/\eps))$, and our new  \Cref{thm:apx-degree} gives a polynomial of degree $O(\ell^{k/k+1} \log^{2/k+1}(1/\eps)).$\footnote{In fact our polynomial has one-sided error (whenever the CNF formula outputs 1, the output of our polynomial is also exactly 1), but this property will not be important for us.}
While these bounds superficially appear  similar to each other, the underlying proof techniques are quite different.  The approximator of \cite{Sherstov20} is built using a recursive approach based on composing an ``outer'' approximator for the $\mathsf{And}$ function with high-accuracy ``inner'' approximators for width-$(k-1)$ CNFs. Very roughly speaking, the recursive nature of the construction requires the ``inner'' approximators to have higher and higher accuracy at deeper and deeper levels of the recursion, and this need for high accuracy ultimately results in the $2^{k/2}$ factor that is present in the \cite{Sherstov20} bound. In contrast, as described below our approach only performs one stage of approximation ``at the end'', so it suffers no such recursive penalty.

\medskip

\noindent {\bf Our approach:}  At a high level, our approach is inspired by quantum walk algorithms  of  Szegedy and Magniez et al.~\cite{Szegedy04,MNRS11} that search for a ``marked set''.
For a width-$k$ CNF formula $g$ over $\zo^\ell$, a crucial object for our approach is the \emph{Johnson graph} $J(\ell,r)$ (for a suitable choice of $r > k$) whose vertices correspond to size-$r$ subsets of $\ell$ and whose edges correspond to subsets that differ in exactly one element.  A useful intuition is that when the walk is at a particular $r$-element subset $R \subseteq [\ell]$, this corresponds to an algorithm ``inspecting'' the values of the coordinates in $R$. 
The ``inner'' part of the polynomial $q(x)$ that we construct corresponds to a random walk on the Johnson graph, which is set up in such a way that the walk is ``killed'' at a vertex $R$ if $R$ contains all the variables of some clause of $g$ that is falsified by $x$ (and thus $R$ causes $g(x)$ to be 0). Vertices $R$ that ``witness'' that $g(x)=0$ in this way correspond to the ``marked sets'' that we would like our walk to hit (if they exist).  Observe that since adjacent vertices of the Johnson graph differ in only one element of $[\ell]$ (i.e.~one variable), an $m$-step walk corresponds to a polynomial of degree $r+m$ (where the $r$ is because of the initial $r$-element vertex at which the walk starts).  So the larger $r$ is, the larger the fraction of vertices that will be ``marked'', but larger $r$ also makes our approach more expensive (i.e., make the polynomial have higher degree).

The natural approach to analyzing how quickly a naive random walk will reach a point in $S$ is by analyzing the eigenspectrum of $J(\ell,r)$ and in particular the spectral gap between the second largest eigenvalue and the largest eigenvalue. This analysis, which is sketched in \Cref{sec:approx-degree}, shows that --- not surprisingly --- a naive random walk leads to a polynomial of degree at least $\ell$, which is a trivial and useless bound.  The crux of our improvement, and the key idea imported from the quantum walk literature, is to \emph{apply a (shifted and scaled) version of the Chebyshev polynomial as the ``outer'' part of our  construction to amplify the spectral gap of the random walk}. (For intuition, using a naive powering of the relevant transition matrix instead of applying the Chebyshev polynomial would correspond instead to a standard random walk.)  We remark that the Chebyshev polynomial is applied only once; this is the ``one stage of approximation'' that was alluded to earlier in this section.  The Chebyshev polynomial enables a ``square root savings'' in degree; a bit more precisely, in our context the savings is from requiring $r/\eps_0$ steps of a naive random walk (where here $\eps_0$ captures the fraction of vertices $R$ which are ``marked'') to requiring only $\sqrt{r/\eps_0}$ steps when the Chebyshev polynomial is used (see \Cref{sec:approx-degree} for all details).    After carefully optimizing parameters, this is the savings that results in the bound of \Cref{thm:apx-degree}.

\subsection{The warmup result:  learning in time $2^{n - 2^{-O(\gamma)}n}$}
\label{sec:tech-overview-warmup}

Let us begin by elaborating on the ``useful way'' (mentioned earlier) in which the CNFs $g_1,\dots,g_s$ typically simplify under the random restriction $(\bL,\bz)$.
The crucial notion turns out to be the notion of a \emph{width-reducing set}, which we denote $Y(L,z)$, that is induced by $L$ and $z$ over the subcube $z \circ \zo^{L}$.  Intuitively, this is a subset $Y(L,z) \subseteq [L]$ with the property that \emph{every clause that survives in every restricted CNF $g_i|_{L,z}$ has at most $k$ literals outside of $Y(L,z)$}; in other words, every assignment to the variables in $Y(L,z)$\footnote{Note that such an assignment can be viewed as a second round of restriction after the variables in $[n] \setminus L$ have already been assigned values according to $z$.} causes each CNF $g_1|_{L,z},\dots,g_s|_{L,z}$ CNF to have width at most $k$. (This is why $Y(L,z)$ is called a ``width-reducing set.'')

We give a probabilistic analysis showing that for a random restriction $(\bL,\bz)$ (where $\bL$ is a random subset of $[n]$ chosen by independently including each element with a carefully chosen probability $p=2^{-O(\gamma)}$, and $\bz$ is uniform over $\zo^{[n] \setminus \bL}$), a simple greedy procedure succeeds in constructing a not-too-large width-reducing set $Y(\bL,\bz)$ with extremely high probability (where here the value of $k$ is taken to be an absolute constant, independent of $\gamma$, times $\log n$; see \Cref{eq:warmup-parameters}). In particular, a consequence of our analysis  (\Cref{simplecoro1a}) states that with high probability over the (one-time initial) choice of $\bL$, only a $2^{- \ell/6}$ fraction of outcomes of $z \in \zo^{[n] \setminus \bL}$ have $|Y(\bL,z)| > \ell/100$, where $\ell=|\bL|$ (we call these ``hard'' outcomes of $z$).

Now, let $d$ be the degree bound on approximating polynomials for width-$k$ CNF formulas that is provided by \Cref{eq:us}. For any $z$ such that $|Y(\bL,z)| \leq \ell/100$, by combining the above ingredients and setting parameters carefully, it is not difficult to establish the existence of a polynomial $Q_z$ of degree at most $d + \ell/100 \leq \ell/90$
over the variables in $\bL$ such that (i) $Q_z(x)$ is never ``too close to zero'' and the coefficients of $Q_z$ aren't ``too large;'' and (ii) the sign of $Q_z(x)$ agrees with $(-1)^{F|_{\bL,z}(x) +1}$ for every input $x \in \zo^{\bL}$ (see \Cref{lem:PTF-with-margin}).  

Given this, a standard application of the Perceptron algorithm over the feature space of all monomials of degree at most $\ell/90$ can be used to learn $F|_{\bL,z}$ over the subcube $z \circ \zo^{\bL}$.  As discussed earlier, for the ``hard'' subcubes that have $|Y(\bL,z)| > \ell/100$, we give up (after making a number of mistakes which exceeds the Perceptron algorithm's mistake bound, which is how we know when to give up) and use a memorization-based approach. The crucial point is that since only a $2^{- \ell/6}$ fraction of subcubes are ``hard'' ones on which we will need to give up, and we make only $2^{\ell}$ mistakes on each such subcube (since it contains only $2^\ell$ points), the overall approach is able to incur a $2^{-\Omega(\ell)}$ multiplicative runtime savings. Since $\ell=|\bL|\approx pn = 2^{-O(\gamma)}n$, this yields the warmup version of our theorem.

\subsection{The full result: learning in time $2^{n - \Omega(1/\gamma)n}$}
\label{sec:tech-overview-full}

How can we improve the $2^{n - 2^{-O(\gamma)}n}$ running time of the warmup?
A first observation is that our algorithmic approach will take time at least $2^{n-|\bL|}$, since it needs to at least ``touch'' each of the $2^{n-|\bL|}$ subcubes corresponding to having $\bL$ as the live set.  Since $|\bL|$ is roughly $pn$, this motivates taking $p$ to be a larger value than in the warmup. It turns out that the best choice for us is to take $p$ to be a small absolute constant (for concreteness we take $p=2^{-10}$), which leads to needing to take $k=\Theta(\gamma \log n)$ in \Cref{thm:apx-degree}.  When $k$ is this large, though, we have a problem, which is that the degree bound that comes out of \Cref{thm:apx-degree}, while less than $\ell$, is at least $\ell/2$; this means that there are $\Omega(2^\ell)$ ``low-degree'' monomials over the variables in $\bL$, which is too many to get any savings using Perceptron over those monomials as was done in the warmup.

Our solution to this problem exploits the fact that the Perceptron algorithm's performance depends only on the norm of its input examples and the margin of the linear separator, and not on the ambient dimension of the space over which it is run.  Leveraging this fact, we fix the above problems by using two new ingredients that 
go beyond the approach of the warmup:

\begin{enumerate}

    \item {\bf A linear approximator over ``block indicator'' features rather than a low-degree polynomial approximator over standard monomials.}  We partition the set $\bL$ of live variables into disjoint blocks $\bL_1,\dots,\bL_B$ of size at most $O(\gamma)$. This is useful for us because it enables us to define a new set of features which correspond to indicator functions of the input taking a particular assignment over a particular collection of blocks.  By adapting the proof of \Cref{thm:apx-degree} to work with these new features, we are able to prove a version of that result which approximates any width-$k$ CNF as a (not-too-high-weight) linear form over those features. Those features turn out to precisely enable the useful properties of Perceptron mentioned above. However, to get a non-trivial overall bound, we need a second ingredient:

    \item {\bf An exponential moment bound, rather than a tail bound, for the width-reducing set $Y(\bL,\bz).$} To make the analysis go through we require tighter control over the random variable $Y(\bL,\bz)$ than in the warmup.  In the warmup it was sufficient to establish a tail bound showing that $|Y(\bL,\bz)|$ is large only with very small probability, but now we will need to bound the probability that $Y(\bL,\bz)$ achieves a given size for essentially all possible sizes.  This is done by giving an exponential moment bound on the random variable $|Y(\bL,\bz)|$; this bound allows us to control the aggregate number of mistakes (and hence the total running time) across all of the $2^{n-|\bL|}$ copies of Perceptron that the overall algorithm executes.

\end{enumerate}


\section{Preliminaries} \label{sec:prelims}

\subsection{The learning model: online mistake-bound learning}
\label{sec:model}

We follow  \cite{ST17itcs} and present our results in the \emph{online mistake-bound} learning model.
As is well known, this model is equivalent to the model of exact learning from equivalence queries only \cite{Angluin:88}.  
The model is very simple and operates as follows:  Given a \emph{target class} $\calC$ 
of functions from $\{0,1\}^n$ to $\{0,1\}$, let $f$ be an unknown \emph{target function} that is promised to belong to $\calC$.  
The learning process unfolds in a sequence of \emph{trials}; throughout the trials the learning algorithm maintains and updates a \emph{hypothesis function} $h: \{0,1\}^n \to \{0,1\}$ (to be more precise, it maintains a  Boolean circuit that computes $h$).  At the start of each trial,
\begin{itemize}

\item If the hypothesis $h$ is logically equivalent to $f$, i.e.~$h(x)=f(x)$ for all $x \in \{0,1\}^n$, then the learner succeeds and the process stops.  (No running time is incurred for this final trial.)

\item Otherwise, in unit time an arbitrary \emph{counterexample}, i.e.~an input $x \in \zo^n$ such that $h(x) \neq f(x)$, is presented to the learning algorithm, and the learning algorithm may update its hypothesis $h$ before the start of the next trial.

\end{itemize}

The running time of a learning algorithm in this model is simply the worst-case running time until the algorithm succeeds, taken over all $f \in \calC$ and all possible sequences of counterexamples. Note that the running time of each trial includes both the time required to evaluate $h$ on an arbitrary example in $\zo^n$ as well as the time required to update $h$ before the next trial.\footnote{If the time to evaluate $h$ does not count towards the running time of the learning algorithm, then a simple trick of memorizing the example sequence and ``offloading all the computation to the evaluation of $h$'' can learn the unrestricted class of all $\poly(n)$-size circuits in $\poly(n)$ time, see Exercise 2.4 of \cite{KearnsVazirani:94}.}
Thus, the running time of an algorithm in this model can be upper bounded by (mistake bound)$\cdot$(maximum time required to evaluate a hypothesis + maximum time required to update a hypothesis).

The online mistake-bound learning algorithm that we give will be randomized.  We say that a randomized algorithm learns class $\calC$ in time $T(n)$ if for any target function $f \in \calC$, the algorithm succeeds with probability
at least $9/10$ (over its internal coin tosses) after at most $T(n)$. We remark that our algorithm succeeds within the stated time bound with probability at least $1-O(1/n)$, and that it uses randomness only once at the very beginning of the algorithm to randomly select the live set $\bL \subset [n]$ of variables. 

Finally, we recall the well-known fact that that learning results in the online mistake-bound learning model easily yield corresponding learning results in the distribution-free PAC model.  More precisely, if a class $\calC$ is learnable in time $T(n)$ in the online mistake-bound model, then $\calC$ is learnable to confidence $1-\delta$ and accuracy $1-\eps$ in the PAC model in time $T_{\mathrm{PAC}} = O({\frac {T(n)}\eps} \ln ({\frac {T(n)}{\delta}}))$ (see the subsection ``The learning model we consider'' in Section~2 of \cite{ST17itcs}).
Hence the polynomial savings that we achieve in the online mistake-bound model directly carry over to the distribution-free PAC model.

\subsection{Johnson Graph}
\emph{Johnson graphs} are central to our results on the approximate degree of width-$k$ CNFs.

\begin{definition}[Johnson Graph]
The Johnson graph $J(\ell,r)$ is a graph whose vertices correspond to $r$-element subsets of $[\ell]$, where two such subsets are adjacent vertices if their intersection has size $r-1$.
\end{definition}

We will consider random walks on the Johnson graph, where a random step consists of moving to a random neighbor of the current vertex. Hence we will crucially need to understand the spectra of Johnson graphs.

\begin{theorem}[\cite{Brouwer1989}]
\label{thm:spectra}
    Let $M$ denote the transition matrix for the above random walk on $J(\ell,r)$. Then for any $i = 0, 1, \dots, \min(r,\ell-r)$ we have that $M$ has eigenvalue
    \[\frac{(r-i)(\ell-r-i) - i}{r(\ell-r)}\]
    with multiplicity $\binom{\ell}{i} - \binom{\ell}{i-1}$, where $\binom{\ell}{-1} = 0$.
\end{theorem}

As a corollary, we then get
\begin{corollary}
\label{cor:spectra}
    Suppose that $2 \leq r \leq \ell - 1$. Let $M$ denote the transition matrix for the random walk on $J(\ell,r)$ and let $\lambda_1, \dots, \lambda_N$ and $u_1, \dots, u_N$ for $N = \binom{\ell}{r}$ denote its eigenvalues and associated eigenvectors, respectively, with $|\lambda_1| \geq |\lambda_2| \geq \dots$.  Then we have $\lambda_1 = 1$, $u_1 = \vec{1}$, and $|\lambda_2| \leq 1 - \frac{1}{r}$.
\end{corollary}

\begin{proof}
    By \Cref{thm:spectra}, we have that $M$ has eigenvalue $1$ with multiplicity $1$. On the other hand, we can note that
    \[\frac{(r-i)(\ell-r-i) - i}{r(\ell-r)} = 1 - \frac{i(\ell+1-i)}{r(\ell-r)} \]
    Thus, when $i \not = 0$, this is maximized when $i = 1$ at
    \[1 - \frac{\ell}{r(\ell-r)} \leq 1 - \frac{1}{r}. \]
    The quantity is minimized when $i = \min(r,\ell-r)$ at
        \[-\frac{1}{\max(r,\ell-r)}\]
    Thus, $|\lambda_2| \leq 1 - \frac{1}{r}$. Finally, we can verify that since $J(\ell,r)$ is regular, we have $M \vec{1} = \vec{1}$. This implies that $u_1 = \vec{1}$, as desired.
\end{proof}

\subsection{Chebychev Polynomials}
We will use \emph{Chebychev polynomials} as another key tool to prove our bounds on the approximate degree of width-$k$ CNFs.
Let \(T_d\) be the degree-\(d\) Chebyshev polynomial, characterized by
\(T_d(\cos t)=\cos(dt)\) and \(T_d(\cosh t)=\cosh(dt)\); see
\cite[Chapters~1--2]{Rivlin1990Chebyshev}.
The properties we need from the shifted and scaled Chebyshev polynomial are standard:

\begin{lemma}[Shifted and scaled Chebyshev polynomial]\label{lem:simple-chebyshev}
Let \(0<\eta\le1\), \(0<\eps<1\), and let
\(d\ge\ln(2/\eps)/\sqrt\eta\) be an integer.  Then the real
polynomial 
\[
  f(\xi)= \sum_{j=0}^d c_j \xi^j := 
  \frac{T_d\!\left((2\xi+\eta)/(2-\eta)\right)}
       {T_d\!\left((2+\eta)/(2-\eta)\right)}.
\]
has degree $d$ and satisfies
\begin{equation}
    \label{eq:f-properties}
  f(1)=1,
  \qquad
  |f(\xi)|\le\eps\quad\text{for all }-1\le\xi\le1-\eta,
  \qquad
  \text{and~}\sum_{j=0}^d |c_j| \leq 8^d.
\end{equation}
\end{lemma}

\begin{proof}
The map $\xi \mapsto (2\xi+\eta)/(2-\eta)$ in the numerator of $f(\xi)$ sends \([-1,1-\eta]\) onto \([-1,1]\),
where \(|T_d|\le1\), and sends \(1\) to
\((2+\eta)/(2-\eta)\ge1+\eta\).  
The polynomial \(T_d\) is increasing on
\([1,\infty)\), and the elementary estimate
\(\arccosh(1+\eta)\ge\sqrt\eta\) holds for \(0<\eta\le1\).  Hence
\[
  T_d\!\left(\frac{2+\eta}{2-\eta}\right)
  \ge T_d(1+\eta)
  =\cosh\!\bigl(d\,\arccosh(1+\eta)\bigr)
  \ge\tfrac12e^{d\sqrt\eta}
  \ge1/\eps,
\]
by our choice of $d$. This establishes the first two parts of \Cref{eq:f-properties}. 
For the third part (the coefficient bound), 
the Chebyshev polynomial recurrence 
\(T_{d+1}(z)=2zT_d(z)-T_{d-1}(z)\) gives
\[
  L_{d+1}\le2L_d+L_{d-1},
\]
where $L_d$ is the sum of absolute values of the coefficients of $T_d$.
Since the positive characteristic root of the above linear recurrence is \(1+\sqrt2\), induction starting from
\(L_0=L_1=1\) gives \(L_d\le(1+\sqrt2)^d\).  Substituting the affine form
\((2\xi+\eta)/(2-\eta)\), whose two coefficient magnitudes sum to at most
three, multiplies this bound by at most \(3^d\).  The denominator is at least
one, so the resulting coefficient mass of $f$ is at most
\((3(1+\sqrt2))^d<8^d\).
\end{proof}

\subsection{Perceptron}

Recall that the Perceptron Algorithm is an online learning algorithm that maintains a vector $w \in \R^m$ (initially $0^m$) of real weights over the feature space.  Given an input example $x \in \R^m$, the Perceptron Algorithm's hypothesis is $h(x) = \sign(w \cdot x) \in \bits$; if the true label $y \in \bits$ disagrees with $h(x)$, then $w$ is updated via the update rule $w \leftarrow w + y x.$

The Perceptron Convergence Theorem gives a mistake bound when the Perceptron Algorithm is run on a sequence of examples that are linearly separable (using a halfspace that passes through the origin) with a margin:

\begin{lemma} [Perceptron Convergence Theorem, \cite{Novikoff:62}]
\label{lem:perceptron}
Suppose that every labeled example $(x,y)$ that is given to the Perceptron Algorithm satisfies both $\|x\| \leq R$ and $0 < \rho \leq y(w^* \cdot x)$ for some vector $w^*$. Then
the Perceptron Algorithm makes at most \(R^2\|w^*\|^2/\rho^2\) mistakes. \end{lemma}

Note that \Cref{lem:perceptron} has no dependence on the ambient dimension $m$; this will be important for us when we present our full result in \Cref{sec:real}.


\section{Approximate Degree of Width-$k$ CNFs} 
\label{sec:approx-degree}

We start by giving a bound on the approximate degree of width-$k$ CNFs. (By Boolean duality, the same bound holds for the approximate degree of width-$k$ DNFs.)  In particular, our main goal in this section will be to prove the following theorem:

\begin{theorem}
\label{thm:apx-degree}
For any CNF $g$ with clauses of width at most $k$ and $\eps \in (0,0.1)$, there exists a polynomial 
    $q$ of degree $d := O(\ell^{k/k+1} \log^{2/k+1}(1/\eps))$ such that
    if $g(x) = 1$ then $q(x) = 1$ and for all $x \in \zo^\ell$ we have that $|q(x) - g(x)| \leq \eps$. Moreover, the total weight of the coefficients of $q$ is at most $32^d$.
\end{theorem}

This substantially improves on the dependence on $k$ in the best prior result, due to Sherstov \cite{Sherstov20}, which was a degree bound of $O(\sqrt{2}^k \ell^{k/k+1} \log^{1/k+1}(1/\eps))$; the cost of this improvement is a slightly worse dependence on $\eps$, but the difference will not matter for us. Crucially, the above will allow us to get non-trivial bounds on the approximate degree of CNFs of width $O(\log(\ell))$, which will roughly be what we can enforce by a random restriction.

As mentioned in the technical overview, we will prove the theorem by considering random walks on the Johnson graph $J(\ell,r)$ (where as we will see we have $r>k$; the parameter $r$ will be set later). Such a random walk consists of first sampling a random subset $\bR_0$ of $r$ variables and then getting each $\bR_{t+1}$ from $\bR_t$ by sampling $\bR_{t+1}$ as a random set of size $r$ with $|\bR_{t+1} \setminus \bR_t| = 1$. Throughout this process, we would like to output $0$ if at some point in time we find a set $\bR_t$ such that $x|_{\bR_{t}}$ falsifies some clause of $g$ and thereby ``certifies'' that $g(x) = 0$. If we find no such set, we would like to output $1$. 

Given these goals, we define
\[
S:=S(x) = \cbra{\text{all sets $R$ in $\binom{[\ell]}{r}$  such that $x_R$ falsifies some clause of $g$}}.
\]
Clearly $S = \emptyset$ if $g(x) =1$. On the other hand if $g(x) = 0$, then there is at least one falsified clause, so the fraction of all vertices that belong to $S$ is
\begin{equation} \label{eq:epszero}
    \frac{|S|}{\binom{\ell}{r}} \geq \frac{\binom{\ell-k}{r-k}}{\binom{\ell}{r}} =  \prod_{i=0}^{k-1} \frac{r-i}{\ell-i} \geq \left( \frac{r-k}{\ell} \right)^k =: \eps_0.
\end{equation}

The naive approach would now be to argue that after walking for not too many steps, we expect to reach a point in $S$. Since the second eigenvalue of the Johnson graph has magnitude roughly $1-1/r$  (cf. \Cref{cor:spectra}), we would naively expect to take roughly $\Omega(r \eps_0^{-1})$ steps to reach $S$.
However, a simple computation  reveals that $\Omega(r \eps_0^{-1}) \geq \ell$ whenever $k > 2$, so this approach reads all variables and we should not expect to get a non-trivial bound from it. 
In the rest of this section we will see that we can do better
by applying an appropriate Chebychev polynomial (cf. \Cref{lem:simple-chebyshev}) to ``amplify'' the spectrum and get a ``square root'' savings that allows us to detect whether $|S| \geq \eps_0$ or is empty using walks of length $O \left(\sqrt{r\eps_0^{-1}} \right)$. 

More formally, our polynomial will be parameterized by numbers $r$ and $m$ to be set later. We then let  $M_r \in \mathbb{R}^{\binom{[\ell]}{r} \times \binom{[\ell]}{r}}$ denote the transition matrix of our random walk on the Johnson graph, and we let $\Pi_{\overline{S}}: \mathbb{R}^{\binom{[\ell]}{r} \times \binom{[\ell]}{r}}$ denote the orthogonal projection onto $\overline{S}$, the complement of $S$. Let $P_m$ denote the degree-$m$ Chebyshev polynomial with $\eta = \eps_0/r$ from \Cref{lem:simple-chebyshev}. Finally, let $v$ be the unit vector $v := \sqrt{\frac{1}{\binom{\ell}{r}}} \cdot \vec{1} \in \R^{\binom{[\ell]}{r}}$ and define our desired polynomial as
    \[q(x) := v^T P_m \left(\Pi_{\overline{S}} M_r \Pi_{\overline{S}} \right) v. \]
(Note that the above is a function of $x$ as $\Pi_{\overline{S}}$ depends on $x$.) 

We start by arguing that $q(x)$ is a polynomial with the desired coefficient bound.

\begin{lemma}
\label{lem:q-is-poly}
$q(x)$ is a polynomial of degree at most $r + m$. Moreover, the sum of the absolute values of the coefficients of $q$ is at most $32^{r + m}$.
\end{lemma}

\begin{proof}
    We start by noting that $\Pi_{\overline{S}}$ is a diagonal matrix with $\left(\Pi_{\overline{S}}\right)_{R,R} = \Indicator_{\overline{S}}(R)$ for all $R \in \binom{[\ell]}{r}$, where $\Indicator_{\overline{S}}$ is the indicator function for the set $\overline{S}$. We can now expand $q(x)$ as
        \[q(x) =  c_0 + \binom{\ell}{r}^{-1} \cdot \left(\sum_{t = 1}^m c_t \sum_{R_0, R_1, ..., R_t} \prod_{i = 0}^{t} \Indicator_{\overline{S}}(R_i) \cdot \prod_{i = 0}^{t-1} (M_r)_{R_i, R_{i+1}} \right)  \]
    where $c_t$ is the coefficient of $\xi^t$ in $P_m(\xi)$. As only walks $R_0, \dots, R_t$ in the Johnson graph contribute to the product above we have that $\prod_{i=0}^t \Indicator_{\overline{S}}(R_i)$ is simply a $\{0,1\}$-valued function on $r + t$ variables, i.e. those that appear in the union of the sets along the walk. As such, $\prod_{i=0}^t \Indicator_{\overline{S}}(R_i)$ can be written as a polynomial of degree at most $r + m$ with integer coefficients of total magnitude at most $4^{r+m}$.
    
    For the overall bound on the weight of the coefficients, we note that it can be upper bounded by 
        \[4^{r+m} \cdot \binom{\ell}{r}^{-1} \cdot \left( |c_0| + \sum_{t = 1}^m  |c_t| \sum_{R_0, R_1, ..., R_t} \prod_{i = 0}^{t-1} (M_r)_{R_i, R_{i+1}} \right) \leq 4^{r+m} \cdot \sum_{t = 0}^m |c_t| \cdot v^T (M_r)^t v \leq 32^{r+m} \]
    where the inequality follows from the bound on the coefficients $c_t$ in \Cref{lem:simple-chebyshev} and the fact that $M_r$ has operator norm at most $1$ by \Cref{cor:spectra}.
\end{proof}

With \Cref{lem:q-is-poly} in hand, we now turn to prove \Cref{thm:apx-degree}.

\begin{proof}[Proof of \Cref{thm:apx-degree}]
    We will assume throughout that $k = O(\log(\ell))$ as otherwise the degree bound is at most $\Omega(\ell)$ and we can exactly write $g$ as a polynomial of degree $\ell$. We set $r: = 2 \ell^{k/(k+1)} \log^{2/(k+1)}(1/\eps)$ and $m: = \log(2/\eps) \cdot \sqrt{r \eps_0^{-1}}$. We then have that by \Cref{lem:q-is-poly}, $q(x)$ has degree at most
    \begin{align*}
        r + m
        &= 2 \ell^{k/(k+1)} \log^{2/(k+1)}(1/\eps) + \log(2/\eps) \cdot \sqrt{r \eps_0^{-1}} \\
        &= 2 \ell^{k/(k+1)} \log^{2/(k+1)}(1/\eps) + \log(2/\eps) \cdot \sqrt{r \cdot \left( \frac{\ell}{r-k} \right)^k } \\
        &\leq 2 \ell^{k/(k+1)} \log^{2/(k+1)}(1/\eps) + \log(2/\eps) \cdot \sqrt{r \cdot \left( \frac{2\ell}{r} \right)^k } \\
        &\leq 2 \ell^{k/(k+1)} \log^{2/(k+1)}(1/\eps) + 2\sqrt{2} \cdot \ell^{k/(k+1)} \cdot  \log^{2/(k+1)}(1/\eps)
    \end{align*}
    and satisfies the claimed weight bound.
    
    It remains to show that $q(x)$ pointwise approximates $g$ to at most an additive $\pm \eps$ error. To start, note that if $g(x) = 1$, then $S = \emptyset$ and $\Pi_{\overline{S}} = I$. Using \Cref{cor:spectra}, we then conclude
        \[q(x) = v^T P_m(M_r) v = P_m(v^T M_r v) = P_m(1) = 1.\]

    On the other hand, if $g(x) = 0$, then by \Cref{eq:epszero} we have that $\frac{|S|}{\binom{\ell}{r}} \geq \eps_0$. We now wish to bound the spectral norm of $\Pi_{\overline{S}} M_r \Pi_{\overline{S}}$. To do this, consider an arbitrary unit vector $z$: we will show that the quadratic form $z^T \Pi_{\overline{S}} M_r \Pi_{\overline{S}} z$ is bounded. We start by writing
        \[\Pi_{\overline{S}} z = u^\parallel + u^\perp \]
    where $u^\parallel$ denotes the component of $\Pi_{\overline{S}} z$ along $v$ and $u^\perp$ denotes the orthogonal component that lies in $v^\perp$. Note that
    \[\|u^\parallel\| = |v^T \Pi_{\overline{S}} z| \leq \sqrt{\frac{|\overline{S}|}{\binom{\ell}{r}}} \cdot \|z\| \leq \sqrt{1 - \eps_0} \]
by Cauchy Schwartz and the fact that
    \[(\Pi_S v)_R = \sqrt{\frac{1}{\binom{\ell}{r}}} \cdot \Indicator_{\overline{S}}(R).\]

    We can now compute 
    \begin{align*}
      z^T \Pi_{\overline{S}} M_r \Pi_{\overline{S}} z &= (u^\parallel + u^\perp) M_r (u^\parallel + u^\perp) \\
      &= u^\parallel M_r u^\parallel + u^\perp M_r u^\perp \\
      &\leq \|u^\parallel\|^2 + (1 - r^{-1}) \cdot \|u^{\perp}\|^2 \tag{by \Cref{cor:spectra}}\\
      &\leq \|\Pi_{\overline{S}} z\|^2 - r^{-1} \|u^\perp\|^2 \\
      &\leq \|\Pi_{\overline{S}} z\|^2 - r^{-1} (\|\Pi_{\overline{S}} z\|^2 - 1 + \eps_0) \\
      &\leq 1 - r^{-1}\eps_0.
    \end{align*}
    Hence $\Pi_{\overline{S}} M_r \Pi_{\overline{S}}$ has spectral norm at most $1 - r^{-1}\eps_0$. This implies that
        \[P_m(\Pi_{\overline{S}} M_r \Pi_{\overline{S}})\]
    has spectral norm at most $\eps$ by our choice of $m$. Thus, 
    \[|q(x)| = \left| v^T P_m( \Pi_{\overline{S}} M_r \Pi_{\overline{S}} ) v \right| \leq \eps \|v\|_2^2 = \eps\]
    as desired, and the proof is complete.
\end{proof}


\section{Warm-up: Learning in Time $2^{n \left(1 - 2^{-O(\gamma)} \right)}$} \label{sec:warmup}

Recall that $s$ is the size of the depth-$3$ target $\mathsf{Or} \circ \mathsf{And} \circ \mathsf{Or} $ circuit, where
  $s=n^\gamma$ and $\gamma$ is a sufficiently large constant.
We use $F$ to denote the unknown Boolean function computed by the circuit, with 
  $F=\lor_{i\le s} g_i$ being the disjunction of CNFs $g_i$, and we 
  use $\calC$ to denote the collection of at most $s$ distinct clauses across all of the $g_i$'s.

\subsection{Subcubes and width-reducing sets} \label{sec:subcubes-warmup}

In both this and the next sections, like the \cite{ST17itcs} algorithm our  online learning algorithm starts by randomly drawing 
  a set of \emph{live} variables $\bL\subset [n]$.
Looking ahead, $\bL$ will be a uniformly random size-$\ell$ subset of $[n]$,
  with $\ell\approx pn$ for some parameter $p$ that we will choose later. (We will see that the choices of $p$ for this section and the next section are quite different.)

For now, consider a fixed set $L$ of live variables of size $\ell$. 
We use it to divide 
$\{0,1\}^n$ into $2^{n-\ell}$ \emph{subcubes}: Given any assignment $\smash{z\in \{0,1\}^{[n]\setminus L}}$, the $z$-subcube, denoted $z \circ \zo^{L}$, contains all $x\in \{0,1\}^n$ with $x_{[n]\setminus L}=z$.
We refer to $z$ as the \emph{address} of the subcube.
We write $F|_{L,z}$ to denote the restriction of $F$ to the $z$-subcube:
$F|_{L,z}(y)=F(y\circ z)$ for any $y\in \{0,1\}^{L}$.
Any clause $C\in \calC$, on the other hand, is either \emph{satisfied}, \emph{falsified} or \emph{surviving} 
  in the $z$-subcube:
\begin{flushleft}\begin{enumerate}
\item $C$ is satisfied if it has a literal already satisfied by $z$;
\item $C$ is falsified if all of its literals are already falsified by $z$; or
\item $C$ is surviving otherwise, meaning that every fixed literal of $C$ in $[n]\setminus L$ is set to be false by $z$ and at least one literal of $C$ lies in the live set $L$.
\end{enumerate}\end{flushleft}
For any surviving $C$ in the $z$-subcube, we write $C|_{L,z}$ to denote the disjunction of its literals in $L$.

Let $k$ be a positive integer parameter to be specified later.
Given $L$ and an address $\smash{z\in \{0,1\}^{[n]\setminus L}}$, we now define the \emph{width-reducing set} $Y(L,z)\subseteq L$. As mentioned in \Cref{sec:tech-overview}, this is called a ``width-reducing set'' because fixing the variables in $Y(L,z)$ in any possible way causes the width of each restricted CNF $g_i|_{L,z}$ to become at most $k$. It is constructed in the following simple greedy fashion:
\begin{flushleft}\begin{enumerate}
\item Start by setting $Y=Y(L,z)$ to be the empty set;
\item Repeat the following: If there remains any surviving clause $C\in \calC$ such that $C|_{L,z}$ contains more than $k$ literals outside of $Y$, pick any such clause $C$ (e.g., by fixing an arbitrary ordering of clauses of $\calC$ and choosing the first such clause) and add all variables in $C|_{L,z}$ outside of $Y$ to $Y$.
\end{enumerate}\end{flushleft}

\noindent The crucial property of $Y(L,z)$, which is evident from how it is constructed, is this:
\begin{equation}
\label{eq:crucialYproperty}
\text{Every $C\in \calC$ that survives in the $z$-subcube has at most $k$ literals outside of $Y(L,z)$.}
\end{equation}

\def\ALG{\textsc{Alg}}

\subsection{Overview of the algorithm and its analysis} \label{sec:warmup-overview}

We will use the following three parameters for the online algorithm \ALGWARMUP in this section:
\begin{equation}
k:=\left\lceil c_0\log n\right\rceil,
\quad
p:=2^{-2\gamma/c_0}, 
\quad \text{and}\quad
\ell:=\lfloor pn \rfloor,
\label{eq:warmup-parameters}
\end{equation}
where $c_0<1$ denotes a sufficiently small positive constant,
  independent of $\gamma$, such that $2^{-1/c_0}$ is sufficiently small compared to the constant hidden in the degree upper bound of  \Cref{thm:apx-degree}.
\ALGWARMUP starts by drawing a uniformly random, size-$\ell$ subset $\bL \subset [n]$ as the set
  of live variables.

In \Cref{sec:tailbound-warmup}, we show that with high probability over the randomness of $\bL$, 
  most of the $2^{n-\ell}$ many $z$-subcubes $z \circ \zo^{\bL}$ satisfy $|Y(\bL,z)|\le \ell/100$. 
(We will refer to such a $z$ as a \emph{tame} address; an address $z$ which is not tame is said to be \emph{hard}.)
This is done by proving a tail bound on the probability of $|Y(\bL,\bz)|>\ell/100$, when the address $\bz$ is 
  drawn uniformly at random from $\smash{\{0,1\}^{[n] \setminus \bL}}$.
Next in \Cref{sec:tameaddress}, we use \Cref{thm:apx-degree} to show~that for every tame address $z$, both the number of mistakes and the
counterexample-processing work of running Perceptron on $F|_{\bL,z}$
  can be explicitly bounded.
Of course, the algorithm \ALGWARMUP cannot tell in advance whether a given address $z$ is tame or hard.
Instead, after drawing $\bL$, it  runs a separate copy of Perceptron on $F|_{\bL,z}$ over each of the $\smash{2^{n-\ell}}$ many $z$-subcubes $z \circ \zo^{\bL}$.
For each $z \in \zo^{[n]\setminus \bL}$, if
 the number of mistakes made by Perceptron on $F|_{\bL,z}$ exceeds the explicit mistake bound for tame addresses, then $\ALGWARMUP$ switches to 
 a brute-force memorization approach to learn $F|_{\bL,z}$ over $\zo^{\bL}.$
We analyze the overall performance of $\ALGWARMUP$ in \Cref{sec:analysis}.

\subsection{Tail bound on the size of the width-reducing set} \label{sec:tailbound-warmup}

For the analysis in this subsection, we consider drawing $\bL$ by including each variable independently with probability $p$, and then drawing a $\smash{\bz\sim\{0,1\}^{[n] \setminus \bL}}$.
The width-reducing set $Y(\bL,\bz)$ is constructed in the same greedy fashion as described earlier, even though
  $\bL$ is not necessarily of size exactly $\ell$.

We prove the following lemma:

\begin{lemma}\label{lem:tailbound}
The probability of $|Y(\bL,\bz)|>\ell/{100}$ is at most $2^{-\ell/5}$.
\end{lemma}

Before proving \Cref{lem:tailbound}, note that given that $\bL$ has size exactly $\ell$ with probability $\Omega(1/n)$, we have from \Cref{lem:tailbound} that when
  $\bL$ is a uniformly random size-$\ell$ subset of $[n]$ and $\smash{\bz\sim \{0,1\}^{[n] \setminus \bL}}$,
  the probability of $|Y(\bL,\bz)|>\ell/100$ is at most $O(n)\cdot 2^{-\ell/5}$. 
We then have the following corollary by Markov's inequality, which we will use later and is the main take-away from this subsection:

\begin{corollary}\label{simplecoro1a}
Let $\bL$ be a uniformly random size-$\ell$ subset of $[n]$.
With probability at least $1-1/n$, the number of $\smash{z\in \{0,1\}^{[n] \setminus \bL}}$ with
  $|Y(\bL,z)|>\ell/100$ is at most $2^{n-\ell}\cdot O(n^2)\cdot 2^{-\ell/5}\le 2^{n-\ell-\ell/6}$.
\end{corollary}

\def\bsigma{\boldsymbol{\sigma}}

We prove \Cref{lem:tailbound} in the rest of this subsection.
To this end, consider any fixed nonempty sequence $A=$ $(C_1,\ldots,C_m)$ of clauses from $\calC$ that can be chosen by the greedy algorithm (in that order) to build $Y(\bL,\bz)$. 
Let
\[
    \Delta_j:=\operatorname{vars}(C_j)
       \setminus\bigcup_{i<j}\operatorname{vars}(C_i),
\]
and let $\bsigma_j$ denote the number of live variables in $\Delta_j$, i.e.~$\bsigma_j := |\bL \cap \Delta_j|$.
For $A$ to be the sequence of clauses chosen by the greedy algorithm to build $Y(\bL,\bz)$, a necessary condition is that for every $j\in [m]$, (1) $|\Delta_j|\ge k+1$; (2) $\bsigma_j \geq k+1$, and 
(3) every variable in $\Delta_j$ that is not in $\bL$ must receive the value in $\bz$ that falsifies~its literal~in $C_j$.
Letting $\beta=(1-p)/2$,
the probability that $\bsigma_j$ takes a particular value $\sigma_j$ is given by
$$
{|\Delta_j| \choose \sigma_j}
p^{\sigma_j} \beta^{|\Delta_j|-\sigma_j}
\le \frac{2}{1+p} \left(\frac{2p}{1+p}\right)^{\sigma_j}\le (3p)^{\sigma_j},
$$
where in the first inequality we used
$$
{|\Delta_j| \choose \sigma_j} \beta^{|\Delta_j|-\sigma_j}
\le \sum_{i\ge \sigma_j} {i\choose \sigma_j} \beta^{i-\sigma_j}
= (1-\beta)^{-\sigma_j-1}=\left(\frac{2}{1+p}\right)^{\sigma_j+1}.
$$
By the disjointness of $\Delta_1,\Delta_2,\dots$ we have that $\bsigma_1,\bsigma_2,\dots$ are independent, and hence
for each fixed nonempty sequence $A=(C_1,\ldots,C_m)$, we have 
\begin{align*}
&\Pr_{\bL,\bz}\big[A\text{ is chosen and }|Y(\bL,\bz)|>\ell/100\big] \\
&\qquad\le \sum_{\substack{\sigma_1,\ldots,\sigma_m\ge k+1\\ \sigma_1+\cdots+\sigma_m\ge \ell/100}} \left(\prod_{j=1}^m \left(3p\right)^{\sigma_j}\right) \tag{union bound}\\
&\qquad\le 2^{-20(\ell/100)}\cdot \sum_{\sigma_1,\ldots,\sigma_m\ge k+1} \left(\prod_{j=1}^m \left(2^{20}\cdot 3p\right)^{\sigma_j}\right)\\
&\qquad= 2^{-\ell/5} \prod_{j=1}^m \left(\sum_{\sigma_j\ge k+1} \left(2^{20}\cdot 3p\right)^{\sigma_j}\right)
\le 2^{-\ell/5} \left(2\cdot \left(2^{20}\cdot 3p\right)^{k+1}\right)^m,
\end{align*}
where we used in the last inequality that $\gamma$ is sufficiently large so $p$ is sufficiently small and hence the sum in the last line is dominated by a geometric series.

Given that there are at most $s^m$ possible sequences $A$ of length $m$, we have
$$
\Pr_{\bL,\bz}\big[|Y(\bL,\bz)|>\ell/100\big]\le 2^{-\ell/5} \sum_{m\ge 1} s^m \left(2\cdot \left(2^{20} \cdot 3p\right)^{k+1}\right)^m\le 2^{-\ell/5},
$$
using that $\gamma$ is sufficiently large so that $2s\cdot (2^{20}\cdot 3p)^{k+1}\le 1/2$
(recalling \Cref{eq:warmup-parameters} and the fact that $s=n^\gamma$; note that here we are using the choice of $p = \exp(-\Theta(\gamma)).$ 
This finishes the proof of \Cref{lem:tailbound}.

\subsection{Perceptron on tame subcubes}\label{sec:tameaddress}

Fix a live set $L$ of size $\ell$ and any tame address $z\in \{0,1\}^{[n]\setminus L}$
  with $|Y(L,z)|\le \ell/100$.
We first use \Cref{thm:apx-degree} to show that there is a low-degree polynomial $Q_z$ that sign-represents
  $F|_{L,z}$ with a margin as follows.
  
\begin{lemma} \label{lem:PTF-with-margin}
There is a polynomial $Q_z$ such that $\deg(Q_z)\le \ell/90$, such that the sum of the squares of the coefficients of $Q_z$ are at most $2^{\ell/20 + o(\ell)}$
and 
$$
(-1)^{F|_{L,z}(x)+1}\cdot Q_z(x)\ge \frac{1}{4}\quad \text{for all $x\in \{0,1\}^L$.}
$$
\end{lemma}
\begin{proof}
Let  $Y= Y(L,z)$ with $|Y|\le \ell/100$.
After the restriction by $(L,z)$, the function  $F|_{L,z}$ can be written as $\lor_{i\le s} h_i$, where 
  each $h_i$ is a CNF over live variables in $L$ such that 
  every clause has at most $k$ literals outside of $Y$. 
Fix an  $h_i$ and any assignment $w\in \{0,1\}^Y$. 
Then $h_i|_{Y,w}$ is a width-$k$ CNF over $L\setminus Y$.
Setting $\eps:=1/(4s)$ in \Cref{thm:apx-degree} gives a polynomial $\Phi_{i,w}$ 
  satisfying
$$
\deg\left(\Phi_{i,w}\right)\le d :=O\left(\ell^{k/(k+1)}\cdot \log^{2/(k+1)}(4s)\right),
\quad\text{and}\quad 
\big|h_{i}|_{Y,w}(x)-\Phi_{i,w}(x)\big|\le \eps$$ 
for all $x\in \{0,1\}^{L\setminus Y}$. 
Given that $s=n^\gamma$ and $k=\Omega(\log n)$, we always have $\log^{2/(k+1)}(4s)\le 2$
  when $n$ is asymptotically large.
On the other hand, we have
$$
\ell^{k/(k+1)}=\ell\cdot \ell^{-1/(k+1)}
\leq \ell \cdot \pbra{{\frac 2 {p n}}}^{1/(k+1)} \leq 2\ell \cdot n^{-1/(k+1)}.
$$
Making $c_0$ a sufficiently small constant to overcome the hidden constant in $d$,
  we have $d\le \ell/1000$.

Next let $\Phi_i$ be the following polynomial over variables in $L$:
$$
\Phi_i(x)=\sum_{w\in \{0,1\}^Y} \Indicator[x_Y=w]\cdot \Phi_{i,w}(x_{L\setminus Y}).
$$
We have $\deg(\Phi_i)\le |Y|+d\le (\ell/100)+(\ell/1000)\le \ell/90$. 
We also have $|\Phi_i(x)-h_i(x)|\le \eps$ for all $x\in \{0,1\}^L$.
Finally, let $Q_z$ be the following polynomial over variables in $L$:
$$
Q_z(x)=\sum_{i\le s} \Phi_i(x)-1/2.
$$
It is easy to verify that $\deg(Q_z)\le \ell/90$ and that $(-1)^{F|_{L,z}+1}\cdot Q_z(x)\ge 1/4$ for all $x$. To bound the weight of the coefficients, we note that $\Indicator[x_Y = w]$ can be written as a polynomial with weights at most $2^{|Y|}$. Using the weight bound from \Cref{thm:apx-degree}, we can then bound the sum of the squares of the coefficients of $Q_z$ by
\[\left( 1/2 + \sum_{i \leq s} \sum_{w\in \{0,1\}^Y} 2^{|Y|} \cdot 32^d \right)^2 = \left( 1/2 + s \cdot 4^{|Y|} \cdot 32^d \right)^2 \leq 2^{\ell/20 + o(\ell)}. \qedhere\]

\end{proof}

Since there are at most $2^{H(1/90)\ell}$ many multilinear monomials of degree at most $\ell/90$ over a space of $\ell$ variables, it follows from \Cref{lem:perceptron} that the Perceptron algorithm, run with target function $F|_{L,z}$ over subcube $z \circ \zo^{L}$ using the feature expansion with a feature $(-1)^{M}$ for each multilinear monomial $M$ of degree at most $\ell/90$, makes at most
$$
2^{H(1/90)\ell+O(\log \ell)}\cdot 2^{\ell/20 + o(\ell)} \le 2^{0.14\ell}
$$  
many mistakes, and each counterexample round has hypothesis update processing time at most $2^{0.14\ell}$ as well.

\subsection{Performance of the online algorithm}\label{sec:analysis}

Assume that the size-$\ell$ subset $\bL$ drawn at the beginning of the algorithm is such that the number of hard addresses is at most $2^{n-\ell-\ell/6}$, which happens with
  probability at least $1-1/n$ by \Cref{simplecoro1a}. 
Then the total number of mistakes
is at most
$$
2^{n-\ell}\cdot 2^{0.14\ell} + 2^{n-\ell-\ell/6}\cdot 2^\ell \le 2 \cdot 2^{n-\ell/6},
$$
and the total running time is bounded by
\[\overbrace{2^{0.14 \ell}}^{\text{hypothesis evaluation time + hypothesis update time}}\cdot \overbrace{2 \cdot 2^{n-\ell/6}}^{\text{mistake bound}} \leq 2^{n - \ell/50}.\]
Using our choice of $\ell$, then yields the desired bound of $2^{n(1-2^{-O(\gamma)})}$ on the runtime and completes the proof of the warm-up case of \Cref{thm:main}.


\section{Learning Depth $3$ Circuits in Time $2^{n \left(1 - \Omega(1/\gamma) \right)}$}
\label{sec:real}

\subsection{Overview of algorithm and its analysis} \label{sec:overview-real}

As in the previous section, our algorithm will restrict to a set $\bL$ of live variables , and for each address $z \in \zo^{[n] \setminus \bL}$ we will then run the Perceptron algorithm over the subcube $z \circ \zo^{\bL}$. That said, we will make two tweaks to get our stronger quantitative result:

\begin{flushleft}\begin{enumerate}
    \item First, we will apply a milder random restriction and handle CNFs of width at most
    \[k := \left\lceil \frac{\gamma}{2}\log(n)\right\rceil \]
    which entails setting
    \[
p := {\frac {1}{2^{10}}}\quad\text{and}\quad
\ell := \lfloor p n \rfloor.
\]  
  \item With this milder restriction, our degree bound, while non-trivial, will be too expensive for the Perceptron algorithm. In particular, suitably modifying \Cref{thm:apx-degree} can only yield approximate degree $d$ with $d \in [\ell/2, \ell]$  when the width of the CNFs may be as large as the value of $k$ given above, which is too high as there are $\Omega(2^\ell)$ monomials of degree at most $\ell/2$ over the $\ell$ variables in $\bL$. To circumvent this, we partition $\bL$ into disjoint ``blocks'' $\bL_1, \dots, \bL_B$, all of which have size at most $b$, where we will set
      \[b := \lceil 8 \gamma \rceil \quad\text{and} \quad B := \left \lceil \frac{\ell}{b} \right \rceil. \]
Let $\bL_U:= \cup_{i\in U} \bL_i$ for any $U\subseteq [B]$. We run the Perceptron algorithm over a set of features $\calF$ given by the feature map $\Phi$:
    \[\Phi_{\emptyset}(x)  = 1 \quad \text{and} \quad \Phi_{U,y}(x) = 1 \big[ x|_{\bL_U} = y \big] \text{~for each nonempty~}U \subseteq [B], y \in \zo^{\bL_U}.
    \]
We remark that the feature map $\Phi$ has more than $2^\ell$ coordinates (features), but it has $2^B$ non-zero entries, each with value 1, for any given $x$. The fact that $\|\Phi(x)\|_2^2 = 2^B$, which is much less than $2^\ell$, will allow us to efficiently run the Perceptron algorithm over these features.
\end{enumerate}\end{flushleft}

Note that throughout this section, we assume without loss of generality that we are working with depth-$3$ circuits of size $n^\gamma$, where $\gamma$ is a sufficiently large absolute constant. Since our bounds are asymptotic in $n$, we can and will also assume that $\gamma \leq \log(n)$ (we will use this later).

\subsection{Dictionary version of approximating polynomial}

We begin by reframing \Cref{thm:apx-degree} in terms of blocks of variables and setting our parameters for the large width setting. Similar to above, for a set of blocks $V_0, \dots, V_{B'}$, we will use the notation $V_U$ for $\bigcup_{i \in U} V_i$ for a set $U \subseteq [B']$. In \Cref{sec:linsep} we will explain how $V_0,\dots,V_{B'}$ relate to the sets $L_0,\dots,L_B$ described above. 
\begin{theorem}
\label{thm:block-apx-degree}
Suppose that 
$[\ell] = V_0 \sqcup V_1 \sqcup V_2 \sqcup \dots \sqcup V_{B'}$ and let $g$ be any CNF over $\zo^\ell$ whose clauses each depend on at most $k$ of the sets $V_1,\dots,V_{B'}$, with $k \in [\log(B'), \sqrt{B'}/4]$. For any $\eps \in (0,0.1]$ and $$\tau = B' - \frac{B' \log(B')}{4k} + 10 (B')^{0.9} \log\left(\frac{1}{\eps}\right),$$ there is a value $a_0$ and functions $h_{V_T \cup V_0}$ such that
    \[q(x) := a_0 + \sum_{\substack{T \subseteq [B']: \\ |T| \leq \tau}} h_{V_T \cup V_0}(x) .\]
Here $h_{V_T \cup V_0}$ are functions  that depends only on variables in $V_T \cup V_0$ and satisfy
\[ |a_0| + \sum_{\substack{T \subseteq [B']: \\ |T| \leq \tau}} \| h_{V_T \cup V_0}\|_{\infty} \leq \exp \left( 10 (B')^{0.9} \log(1/\eps) \right),\]
where $\|h_{V_T \cup V_0}\|_\infty$ denotes $\max |h_{V_T \cup V_0}(x)|$ over all $0/1$ assignments to the variables in $V_T \cup V_0$.
Moreover, we have that $g(x) = 1$ implies that $q(x) = 1$ and that $|g(x) - q(x)| \leq \eps$ for all $x \in \zo^\ell$.
\end{theorem}

We prove \Cref{thm:block-apx-degree} analogously to \Cref{thm:apx-degree}, but taking walks on sets of blocks, i.e. 
using the Johnson graph with vertex set $\binom{[B']}{r}$.

\begin{proof}
    We again use parameters $r$ and $m$. Similar to before, we will consider the Johnson graph over subsets of blocks. Let $M_r$ denote the transition matrix on the Johnson graph and let $S=S(x) \subseteq \binom{[B']}{r}$ denote the set of sets of blocks such that $x|_{V_S \sqcup V_0}$ falsifies some clause in $g$. As in \Cref{thm:apx-degree}, we  observe that
        \[{\frac {|S|}{\binom{B'}{r}}} \geq \frac{\binom{B'-k}{r-k}}{\binom{B'}{r}} \geq \left( \frac{r-k}{B'} \right)^k := \eps_0.\]
    We then set $r = \left \lfloor B' - B' \log(B')/(4k) \right \rfloor$. Taking $m = \left \lceil \log(2/\eps) \sqrt{r /\eps_0} \right \rceil$, we can observe that
    \begin{align}
    m&\leq 1 + \log(2/\eps) \sqrt{B' \left(1 - \log(B')/(4k) \right) \cdot \left( \frac{B'}{r-k} \right)^k} \nonumber \\
    &\leq 1 + \log(2/\eps) \sqrt{B' \cdot \left( \frac{1}{1 - \log(B')/(4k) - (k+1)/B'} \right)^k} \nonumber \\
    &\leq 1 + \log(2/\eps) \sqrt{B' e^{2k \left(\log(B')/(4k) + (k+1)/B' \right)}} \nonumber \\
    &\leq 2 \sqrt{e} \log(1/\eps) (B')^{0.9}  \label{eq:peanut}
    \end{align}
    where in the penultimate inequality we used the fact that $1 - \log(B')/(4k) - k/B' \geq \frac{1}{2}$ by our assumptions on $k$ and that $1 - x \geq e^{-2x}$ for $x \in [0,1/2]$.
    
    We now let $P_m$ denote the degree $m$ polynomial from \Cref{lem:simple-chebyshev} with $\eta = \eps_0/r$ and  $v = \binom{B'}{r}^{-1/2} \cdot \vec{1} \in \R^{{[B'] \choose r}}$, and we consider the expression
    \[q(x) := v^T P_m(\Pi_{\overline{S}} M_r \Pi_{\overline{S}}) v.\]

    Expanding this product, we get that
        \[q(x) = c_0 + \binom{B'}{r}^{-1} \cdot \left( \sum_{t = 1}^m c_t \sum_{R_0, R_1, ..., R_t} \prod_{i = 0}^{t} 1_{\overline{S}}(R_i) \cdot \prod_{i = 0}^{t-1} (M_r)_{R_i, R_{i+1}} \right) \] 
    where again $c_t$ is the coefficient of $\xi^t$ in $P_m(\xi)$. As before, only walks of length at most $m$ contribute to the sum so each product only depends on $V_0$ and at most $r + m \leq \tau$ other blocks. Thus, we can set
        \[h_{V_T \cup V_0}(x) = \binom{B'}{r}^{-1} \cdot \sum_{t = 1}^m c_t \sum_{\substack{R_0, R_1, ..., R_t:\\ \bigcup_{i=0}^t R_i = T}} \prod_{i = 0}^{t} 1_{\overline{S}}(R_i) \cdot \prod_{i = 0}^{t-1} (M_r)_{R_i, R_{i+1}}.\]
    and $a_0 = c_0$. We can then compute that
        \begin{align*}
             |a_0| + \sum_{\substack{T \subseteq [B']: \\ |T| \leq \tau}} \| h_{V_T \cup V_0}\|_{\infty} &\leq |c_0| + \binom{B'}{r}^{-1} \cdot \left( \sum_{\substack{T \subseteq [B']: \\ |T| \leq \tau}} \left \| \sum_{t = 0}^m c_t \sum_{\substack{R_0, R_1, ..., R_t:\\ \bigcup_{i=0}^t R_i = T}} \prod_{i = 0}^{t} 1_{\overline{S}}(R_i) \cdot \prod_{i = 0}^{t-1} (M_r)_{R_i, R_{i+1}} \right\|_\infty \right) \\
             &\leq |c_0| + \binom{B'}{r}^{-1} \cdot \left( \sum_{\substack{T \subseteq [B']: \\ |T| \leq \tau}} \sum_{t = 1}^m |c_t| \sum_{\substack{R_0, R_1, ..., R_t:\\ \bigcup_{i=0}^t R_i = T}} \prod_{i = 0}^{t-1} (M_r)_{R_i, R_{i+1}} \right) \\
             &= \left(\sum_{t = 0}^m |c_t| \cdot v^T (M_r)^t v \right) \\
             &\leq 8^m \leq \exp \left( 10 (B')^{0.9} \log(1/\eps) \right)
        \end{align*}
        where the final line used \Cref{lem:simple-chebyshev} and \Cref{eq:peanut}.

    By an identical argument as in the proof of \Cref{thm:apx-degree}, it also follows that if $g(x) = 1$ then $q(x) = 1$ and $|q(x) - g(x)| \leq \eps$ for all $x \in \zo^{\ell}$.
\end{proof}

\subsection{Moment bound on the size of the width-reducing set} \label{sec:momentbound-real}

As in the analysis of \Cref{sec:tailbound-warmup}, it will be convenient to first consider drawing $\bL \subseteq [n]$ by including each variable independently with probability $p$ (though recall that now $p=1/2^{10}$).  Given an outcome of $L$ and $z \in \zo^{[n]\setminus L}$, the width-reducing set $Y(L,z)$ is defined exactly as in \Cref{sec:subcubes-warmup}, 
so condition \eqref{eq:crucialYproperty} again holds for us:  every $C\in \calC$ that survives in the $z$-subcube has at most $k$
literals outside of $Y(L,z)$.

The key technical tool in this section is the following moment bound (note that in the following lemma $\bL$ is constructed as described above, by independently including each variable with probability $p=2^{-10}$):

\begin{lemma} \label{lem:momentbound}
$\E[2^{5|Y(\bL,\bz)|}] \leq 2.$
\end{lemma}
Before proving \Cref{lem:momentbound}, note that since $\bL$ has size exactly $\ell$ with probability $\Omega(1/n)$, we have from \Cref{lem:momentbound} that for $\bL$ a uniformly random size-$\ell$ subset of $[n]$ and $\smash{\bz\sim \{0,1\}^{[n] \setminus \bL}}$, it holds that
\[
    \Ex_{\bL \sim {[n] \choose \ell}}\sbra{\Ex_{\bz}\sbra{2^{5|Y(\bL,\bz)|}}} \leq O(n).
\]
Rewriting the inner expectation as a sum over all $2^{n-\ell}$ addresses $z$, Markov's inequality gives the following, which we will use later and is the main take-away from this subsection:

\begin{corollary}\label{simplecoro1}
Let $\bL$ be a uniformly random size-$\ell$ subset of $[n]$.
With probability at least $1-1/n$ we have
\begin{equation} 
\label{eq:exact-ell}
\sum_{z \in \zo^{[n] \setminus \bL}} 2^{5|Y(\bL,\bz)|} \leq O(n^2) 2^{n-\ell}.
\end{equation}
\end{corollary}

In the rest of this subsection we prove \Cref{lem:momentbound}. The proof is similar to that of \Cref{lem:tailbound} with just a few changes.
Consider any fixed sequence $A=$ $(C_1,\ldots,C_m)$ of clauses from $\calC$ that can be chosen by the greedy algorithm in that order to build $Y(\bL,\bz)$. 
Let
\[
    \Delta_j:=\operatorname{vars}(C_j)
       \setminus\bigcup_{i<j}\operatorname{vars}(C_i),
\]
and let $\bsigma_j$ denote the number of live variables in $\Delta_j$, i.e.~$\bsigma_j := |\bL \cap \Delta_j|$.
For $A$ to be the sequence of clauses chosen by the greedy algorithm to build $Y(\bL,\bz)$, a necessary condition is that for every $j\in [m]$, (1) $|\Delta_j|\ge k+1$; (2) $\bsigma_j \geq k+1$, and 
(3) every variable in $\Delta_j$ that is not in $\bL$ must receive the value in $\bz$ that falsifies~its literal~in $C_j$.
Let $\bH_j$ denote the compound event  that (1) and (2) and (3) all hold.
By the disjointness of $\Delta_1,\Delta_2,\dots$ we have that $\bH_1,\bH_2,\dots$ are independent, and hence we have 
\[
\E\sbra{
\Indicator\sbra{A\text{~is chosen}}\cdot 2^{5|Y(\bL,\bz)|}} \leq
\prod_{j=1}^m \E\sbra{\Indicator[\bH_j] \cdot 2^{5\bsigma_j}}.
\]
As in the earlier analysis,
the probability that $\bsigma_j$ takes a particular value $\sigma_j$ is at most
$(3p)^{\sigma_j}$,
so we have
\begin{align*}
\E\sbra{\Indicator[\bH_j]\cdot 2^{5\bsigma_j}}
\leq \sum_{\sigma_j > k} (3p)^{\sigma_j} \cdot 2^{5\sigma_j}
&=\sum_{\sigma_j > k} \pbra{\frac 3 {32}}^{\sigma_j} \tag{recalling that $p=2^{-10}$}\\
&= {\frac {32}{29}} \cdot \pbra{{\frac 3 {32}}}^{k+1},
\end{align*}
and hence
\begin{equation}
\E\sbra{
\Indicator\sbra{A\text{~is chosen}}\cdot 2^{5|Y(\bL,\bz)|}} \leq
\pbra{
{\frac {32} {29}} \cdot \pbra{{\frac 3 {32}}}^{k+1}
}^m.
\label{eq:summe}
\end{equation}

Now we are ready to upper bound $\E[2^{5|Y(\bL,\bz)|}]$ by summing \Cref{eq:summe} over all $A=(C_1,\dots,C_m).$ To begin,
the $m=0$ case contributes 1 to $\E[2^{5|Y(\bL,\bz)|}]$ since it corresponds to $|Y(\bL,\bz)|=0$.  For $m \geq 1$ there are at most $s^m$ possible sequences $A$ of length $m$, so all in all we have
\[
\E[2^{5|Y(\bL,\bz)|}]
\leq 1 + \sum_{m \geq 1} s^m 
\pbra{{\frac {32}{29}} \cdot \pbra{{\frac 3 {32}}}^{k+1}}^m \leq 2,
\]
where the last inequality is because our choice of parameters in \Cref{sec:overview-real} ensures that we have ${\frac {32} {29}} s \cdot \pbra{{\frac 3 {32}}}^{k+1} \leq 1/2.$
This finishes the proof of \Cref{lem:momentbound}. \qed

\subsection{A linear separator over feature space ${\cal F}$ for the $z$-subcube} \label{sec:linsep}
Fix a set $L$ and an address $z$. In this section, we will show that over our new feature space $\calF$ that is given by $\Phi$, each restricted function $F|_{L,z}$ corresponds to a large margin separator, where the margin depends on the size of the width-reducing set. Towards this end, we will define $\Gamma$ to be the following function of a parameter $t$:
$$\Gamma(t):=\left(1-\frac{1}{2\gamma}\right) \ell+\frac{b}{2\gamma} t.$$

We will then show

\begin{proposition} \label{prop:separator-for-z-subcube}
There is a weight vector $w^*_z$ over the feature space $\calF$ such that
\begin{flushleft}\begin{itemize}
    \item [(a)] $w^*_z$ defines a linear separator over $\R^{\cal F}$, with margin $1/4$, that agrees with the restricted target function $F|_{L,z}$:
    i.e.~
    \[
(-1)^{F|_{L,z}(x) + 1} \big(w^*_z \cdot \Phi(x)\big) \geq 1/4
\quad 
\text{for every~}x \in \zo^L.
    \]
    \item [(b)] the norm of $w^*_z$ is bounded:  $\|w^*_z\|^2 \leq 2^{\Gamma(|Y(\bL,z)|)+o(\ell)}$.
\end{itemize}\end{flushleft}
\end{proposition}
\begin{proof}
Let $G \subseteq [B]$ denote the set of blocks that touch $Y(L,z)$ and let $t = |G|$, for which we always have $t\le |Y(L,z)|$.
Let $B'=B-t$, i.e., the number of blocks that are not touched by $Y(L,z)$. The analysis is then composed of two cases: when $B'$ is large we will use \Cref{thm:block-apx-degree} to create our linear separator, and when $B'$ is small it will suffice to naively expand $F|_{L,z}$ as a threshold function over our features. 

\medskip

\noindent \textbf{Case 1:} $B'\ge B/\log (n)$.
The function $F|_{L,z}$ can be written as $\lor_{i\le s} (g_i|_{L,z})$,
where each $g_i|_{L,z}$ is a CNF over live variables in $L$ such that every clause contains variables from at most $k$
blocks outside of those of $G$.

We will now form blocks $V_i$ so as to apply \Cref{thm:block-apx-degree} to each $g_i|_{L,z}$. Let $V_0 = L_G$ and $V_1, \dots, V_{B'}$ denote the remaining blocks $L_i$ outside of $G$. By our assumptions on $\gamma$, it follows that $k \geq \log(n) \geq \log(B)$ and $k \leq \sqrt{B'}/4$, since $B' \geq B / \log(n)$.
It follows from \Cref{thm:block-apx-degree}, setting $\eps=1/(4s)$, that each $g_i|_{L,z}$ can be approximated by 
$$
\Psi_i=a_0^{(i)}+
\sum_{U} h_U^{(i)},
$$
where the sum is only over $U\subseteq [B]\setminus G$ of size at most
$$
\tau= B' - \frac{B' \log(B')}{4k} + 10 (B')^{0.9} \log\left(4s\right)
$$
and $h_U^{(i)}$ only depends on variables in $L_{G\cup U}$,
and $\Psi_i$ satisfies $|a_0^{(i)}|+\sum_U \|h_U^{(i)}\|_\infty\le 2^{o(\ell)}$. We can then note that any function in this expansion depends on at most
\begin{align*}
    b(\tau + t) &\leq b \left( B' - \frac{B' \log(B')}{4k} + 10 (B')^{0.9} \log\left(4s\right) + t \right) \\
    &= b\left( B - \frac{(B-t) \log(B')}{4k} \right) + o(\ell) \\
    &\leq \ell - \frac{\ell \log(B')}{4k} + \frac{b \log(B')}{4k} \cdot t + o(\ell) \\
    &\leq \ell - \frac{\ell \log(B')}{4k} + \frac{b}{2 \gamma} \cdot t + o(\ell) \\
    &\leq \ell - \frac{\ell}{2 \gamma} + \frac{b}{2 \gamma} \cdot t + o(\ell) = \Gamma(t) + o(\ell)
\end{align*}
variables in $L$, where the final inequality used $\log(B') = (1 - o(1)) \log(n)$ by our assumption on the size of $B'$ and the fact that $\gamma \leq \log(n)$.

Now set $\Psi=\sum_{i=1}^s \Psi_i-1/2$. Note that we can then write $\Psi$ as a linear function $w_z \cdot \Phi(x)$: To see this, decompose each $h_U^{(i)}$ into
\[h_U^{(i)}(x) = \sum_{y \in \zo^{L_{G\cup U}}} h_U^{(i)}(0^{L \setminus L_{G\cup U}} \circ y) \cdot \Indicator[x_{L_{G\cup U}} = y]\]
where we crucially used that $h_U^{(i)}(x)$  only depends on the coordinates in $G \cup U$. Thus, we can take $w_z^*$ to satisfy
    \[(w_z^*)_\emptyset = - 1/2 + \sum_{i=1}^s a_0^{(i)} \quad \text{ and } \quad (w_z^*)_{G \cup U,y} = \sum_{i=1}^s h_U^{(i)}(0^{L \setminus L_{G\cup U}} \sqcup y) \]
    for any $U\subseteq [B]\setminus G$ of size at most $\tau$.

By the properties of each $\Psi_i$ given by \Cref{thm:block-apx-degree}, if $F|_{L,z}(x) = 1$, then
    \[\Psi(x) \geq \frac{1}{4} \]
and if $F|_{L,z} = 0$ then
    \[\Psi(x) \in [-3/4, -1/4],\]
yielding item $(a)$. 

For item $(b)$, we can compute 
\begin{align*}
    \|w_z^*\|^2 &\leq \left( \frac{1}{2} + \sum_{i=1}^s |a_0^{(i)}| \right)^2 + \sum_{U} 2^{\Gamma(t) + o(\ell)} \cdot \left( \sum_{i=1}^s \left\|h_U^{(i)} \right\|_{\infty} \right)^2 \\
    &\leq \left( \frac{1}{2} + \sum_{i=1}^s |a_0^{(i)}| \right)^2 + 2^{\Gamma(t) + o(\ell)} \cdot \left(\sum_{i=1}^s \sum_{U}  \left\|h_U^{(i)} \right\|_{\infty} \right)^2 \\
    &\leq 2^{\Gamma(t) + o(\ell)} \cdot \left(\sum_{i=1}^s \frac{1}{2} + |a_0^{(i)}| + \sum_{U}  \left\|h_U^{(i)} \right\|_{\infty} \right)^2\\
    &\leq 2^{\Gamma(t) + o(\ell)}\cdot 10 s^2 \exp \left( 20 \ell^{0.9} \log(4s) \right) \\
    &\leq 2^{\Gamma(|Y(L,z)|) + o(\ell)}
\end{align*}
as desired.\medskip

\noindent \textbf{Case 2:} $B' \leq B/\log(n)$.
In this case, $t \geq B(1 - 1/\log(n))$ and thus
\begin{align*}
    \Gamma(t) &= \left (1 - \frac{1}{2\gamma} \right) \ell + \frac{b}{2 \gamma} t  \geq \left (1 - \frac{1}{2\gamma} \right) \ell + \frac{b}{2 \gamma} B\left(1-\frac{1}{\log(n)}\right)  \ge \ell - \frac{\ell}{2\gamma \log(n)}.
\end{align*}
Thus taking the $o(\ell)$ to be sufficiently large i.e. $\omega(\ell/\log(n))$, we then have that the expression in property $(b)$ satisfies $2^{\Gamma(t) + o(\ell)} \geq 2^\ell$. In this regime we can use a trivial encoding of $F|_{L,z}$; namely, let
\[\quad (w_z)_{[B], y} = F|_{L,z}(y) - \frac{1}{2} \qquad \text{for all~}y \in \zo^{L}.\]
Property $(a)$ is then trivially satisfied and the weight is easily bounded by
    \[\|w_z\|^2 \leq \sum_{y \in \zo^{L}} \left(F|_{L,z}(y) - \frac{1}{2}\right)^2 = \frac{1}{4} 2^\ell \leq 2^{\Gamma(t) + o(\ell)} \leq 2^{\Gamma(|Y(L,z)|) + o(\ell)},\]
    as required by property $(b)$.
This finishes the proof of \Cref{prop:separator-for-z-subcube}.
\end{proof}

\subsection{Concluding the proof of the strong form of \Cref{thm:main}}

We now combine the ingredients above to prove \Cref{thm:main} with $c_\gamma = \Omega(1/\gamma)$. As discussed earlier, our algorithm randomly chooses a subset $\bL$ of $\ell$ variables from $x_1,\dots,x_n$, and for each address $z \in \zo^{[n] \setminus \bL}$, it runs a separate copy of Perceptron over the subcube $z \circ \zo^{\bL}$ using the feature map $\Phi.$ 

We start with the total mistake bound summed over all $2^{n-\ell}$ copies of Perceptron. Since $\|\Phi(x)\|_2^2 = 2^B$, we can combine \Cref{prop:separator-for-z-subcube} and \Cref{lem:perceptron} to get that for each address $z$, the copy of Perceptron running over subcube $z \circ \zo^{L}$ makes at most
    \[M_z := 2^B \cdot 2^{\Gamma(|Y(L,z)|) + o(\ell)} \cdot 16 = 2^{\Gamma(|Y(L,z)|) + B + o(\ell)}\]
mistakes. Thus, summing over all addresses, the algorithm makes at most
\begin{align*}
    \sum_{z \in \zo^{n \setminus \bL}} M_z &= \sum_{z \in \zo^{n \setminus \bL}} 2^{\Gamma(|Y(L,z)|) + B + o(\ell)} \\
    &= \sum_{z \in \zo^{n \setminus \bL}} 2^{\left(1 - \frac{1}{2\gamma}\right) \ell + \frac{b}{2\gamma} \cdot |Y(L,z)| + B + o(\ell)} \\
    &= 2^{\left(1 - \frac{1}{2\gamma}\right) \ell + B + o(\ell)} \sum_{z \in \zo^{n \setminus \bL}} 2^{\frac{b}{2\gamma} \cdot |Y(L,z)|} 
\end{align*}mistakes in total. Since $\frac{b}{2\gamma} \leq 5$, \Cref{simplecoro1} then gives that, with high probability, the total number of mistakes is at most 
\[ 2^{\left(1 - \frac{1}{2\gamma}\right) \ell + B + o(\ell)} \cdot O(n^2) 2^{n-\ell}  
= 2^{n - \frac{\ell}{2\gamma}  + B + o(\ell)}.
\]

We will now use this mistake bound to bound the total runtime of our algorithm as defined in \Cref{sec:model}.
In each execution of a copy of the Perceptron algorithm, for each input example $x$ we compute only the $2^B$ nonzero features of the feature expansion $\Phi(x)$, and across the entire run of that copy of Perceptron we store only the non-zero entries of its hypothesis vector $w.$ When we make a mistake, we update the entries of $w$ by going through each non-zero entry in $y \cdot \Phi(x)$ and adding it to $w$. So given an input example $x$ for one copy of Perceptron, computing the $2^B$ non-zero entries of $\Phi(x)$, evaluating $w \cdot \Phi(x)$, and updating $w$ if necessary takes time $\poly(n) \cdot 2^B$ using an appropriate data structure.
Thus the total running time across all the copies of Perceptron is bounded by 
\[\overbrace{\poly(n) \cdot 2^B}^{\text{hypothesis evaluation time + hypothesis update time}}\cdot \overbrace{2^{n - \frac{\ell}{2\gamma}  + B + o(\ell)}}^{\text{mistake bound}} \leq 2^{n - \frac{\ell}{6\gamma}} = 2^{n(1 - \Omega(\gamma^{-1}))} \]
as desired.
This concludes the proof of the strong form of \Cref{thm:main}.

\section*{Acknowledgements}
X.C. is supported by NSF grants IIS-1838154, CCF-2106429, and CCF-2107187. S.P. was supported
by NSF grants CCF-2106429, CCF-2107187, CCF-2218677, ONR grant ONR-13533312, and a
NSF Graduate Student Fellowship. R.A.S. is supported in part by NSF awards CCF-2106429 and
CCF-2211238.

\begin{flushleft}
\bibliographystyle{alpha}
\bibliography{allrefs}
\end{flushleft}

\end{document}